\documentclass[11pt,a4paper]{amsart}
\usepackage{amssymb,amsmath}
\usepackage{mathrsfs}
\usepackage{newtxtext}
\usepackage[alphabetic]{amsrefs}
\usepackage{graphicx,color}
\usepackage[utf8]{inputenc}
\usepackage{amsthm}
\usepackage{latexsym}
\usepackage{slashed}
\usepackage[all]{xy}
\usepackage{hyperref}
\usepackage[mathscr]{eucal}
\usepackage{tikz}
\usepackage{mathtools}
\usepackage{accents}
\usepackage{amsaddr}

\def\theequation{\thesection.\arabic{equation}}
\makeatletter
\@addtoreset{equation}{section}
\makeatother

\makeatletter
\newcommand{\eqnum}{\refstepcounter{equation}\textup{\tagform@{\theequation}}}
\makeatother

\newcounter{copy}
\makeatletter
\renewcommand{\thecopy}{\ifnum0=\c@section\arabic{copy}\else\thesection.\arabic{copy}'\fi}
\makeatother

\BibSpec{book}{%
    +{}  {\PrintPrimary}                {transition}
    +{.} { \textit}                     {title}
    +{.} { }                            {part}
    +{:} { \textit}                     {subtitle}
    +{,} { \PrintEdition}               {edition}
    +{}  { \PrintEditorsB}              {editor}
    +{,} { \PrintTranslatorsC}          {translator}
    +{,} { \PrintContributions}         {contribution}
    +{,} { }                            {series}
    +{,} { \voltext}                    {volume}
    +{,} { }                            {publisher}
    +{,} { }                            {organization}
    +{,} { }                            {address}
    +{,} { }                            {status}
    +{,} { \PrintDOI}                   {doi}
    +{,} { \PrintISBNs}                 {isbn}
    +{}  { \parenthesize}               {language}
    +{}  { \PrintTranslation}           {translation}
    +{;} { \PrintReprint}               {reprint}
    +{,} { \PrintDate}                  {date}
    +{.} { }                            {note}
    +{.} {}                             {transition}
}
\BibSpec{article}{%
    +{}  {\PrintAuthors}                {author}
    +{,} { \textit}                     {title}
    +{.} { }                            {part}
    +{:} { \textit}                     {subtitle}
    +{,} { \PrintContributions}         {contribution}
    +{.} { \PrintPartials}              {partial}
    +{,} { }                            {journal}
    +{}  { \textbf}                     {volume}
    +{}  { \PrintDatePV}                {date}
    +{,} { \issuetext}                  {number}
    +{,} { \eprintpages}                {pages}
    +{,} { }                            {status}
    +{,} { \PrintDOI}                   {doi}
    +{}  { \parenthesize}               {language}
    +{}  { \PrintTranslation}           {translation}
    +{;} { \PrintReprint}               {reprint}
    +{.} { }                            {note}
    +{.} {}                             {transition}
}
\BibSpec{collection.article}{%
    +{}  {\PrintAuthors}                {author}
    +{,} { \textit}                     {title}
    +{.} { }                            {part}
    +{:} { \textit}                     {subtitle}
    +{,} { \PrintContributions}         {contribution}
    +{,} { \PrintConference}            {conference}
    +{}  {\PrintBook}                   {book}
    +{,} { }                            {booktitle}
    +{,} { \PrintDateB}                 {date}
    +{,} { pp.~}                        {pages}
    +{,} { }                            {publisher}
    +{,} { }                            {organization}
    +{,} { }                            {address}
    +{,} { }                            {status}
    +{,} { \PrintDOI}                   {doi}
    +{,} { \eprint}        {eprint}
    +{}  { \parenthesize}               {language}
    +{}  { \PrintTranslation}           {translation}
    +{;} { \PrintReprint}               {reprint}
    +{.} { }                            {note}
    +{.} {}                             {transition}
}

\BibSpec{misc}{
  +{}{\PrintAuthors}  {author}
  +{,}{ \textit}     {title}
  +{}{ (}             {date}
  +{),}{ }             {note}
  +{.}{}              {transition}
}

\theoremstyle{definition}
\newtheorem{defn}[equation]{Definition}

\theoremstyle{plain}
\newtheorem{thm}[equation]{Theorem}
\newtheorem{prp}[equation]{Proposition}
\newtheorem{lem}[equation]{Lemma}

\theoremstyle{remark}
\newtheorem{rmk}[equation]{Remark}

\newcommand{\bB}{\mathbb{B}}
\newcommand{\bC}{\mathbb{C}}
\newcommand{\bD}{\mathbb{D}}

\newcommand{\bH}{\mathbb{H}}

\newcommand{\bK}{\mathbb{K}}

\newcommand{\bM}{\mathbb{M}}
\newcommand{\bN}{\mathbb{N}}

\newcommand{\bR}{\mathbb{R}}

\newcommand{\bT}{\mathbb{T}}

\newcommand{\bZ}{\mathbb{Z}}

\newcommand{\cD}{\mathcal{D}}
\newcommand{\cE}{\mathcal{E}}

\newcommand{\cG}{\mathcal{G}}

\newcommand{\cL}{\mathcal{L}}

\newcommand{\cP}{\mathcal{P}}
\newcommand{\cQ}{\mathcal{Q}}

\newcommand{\cS}{\mathcal{S}}

\newcommand{\cU}{\mathcal{U}}

\newcommand{\cW}{\mathcal{W}}

\newcommand{\fD}{\mathfrak{D}}

\newcommand{\fM}{\mathfrak{M}}

\newcommand{\bk}{\mathbf{k}}

\newcommand{\bv}{\mathbf{v}}

\newcommand{\bx}{\mathbf{x}}

\newcommand{\sD}{\mathscr{D}}
\newcommand{\sE}{\mathscr{E}}
\newcommand{\sF}{\mathscr{F}}

\newcommand{\sH}{\mathscr{H}}

\newcommand{\bfI}{\mathbf{I}}

\newcommand{\blank}{\text{\textvisiblespace}}

\newcommand{\id}{\mathrm{id}}

\newcommand{\Cl}{\mathit{C}\ell}
\newcommand{\bCl}{\mathbb{C}\ell}

\DeclareMathOperator{\hotimes}{\hat{\otimes }}
 
\DeclareMathOperator{\K}{\mathrm{K}}
\DeclareMathOperator{\KR}{\mathrm{KR}}
\DeclareMathOperator{\KO}{\mathrm{KO}}

\DeclareMathOperator{\End}{\mathrm{End}}
\DeclareMathOperator{\Hom}{\mathrm{Hom}}
\DeclareMathOperator{\Index}{\mathrm{Index}}

\DeclareMathOperator{\rank}{rank}

\newcommand{\lwedge}{{\textstyle\bigwedge}}

\author{Yosuke Kubota}
\address{Department of Mathematical Sciences, Shinshu University\\ 3-1-1 Asahi, Matsumoto, Nagano, 390-8621, Japan\\ and \\ iTHEMS Program, RIKEN\\ 2-1 Hirosawa, Wako, Saitama, 351-0198, Japan}
\email{ykubota@shinshu-u.ac.jp}
\title[The index theorem of lattice Wilson--Dirac operators]{The index theorem of lattice Wilson--Dirac operators via higher index theory}

\date{September 8, 2020}

\begin{document}
\maketitle
\begin{abstract}
We give a proof of the index theorem of lattice Wilson--Dirac operators, which states that the index of a twisted Dirac operator on the standard torus is described in terms of the corresponding lattice Wilson--Dirac operator.
Our proof is based on the higher index theory of almost flat vector bundles.
\end{abstract}
\tableofcontents
\section{Introduction}\label{section:1}
Lattice gauge theory is a theoretical and numerical approach for studying the Yang--Mills gauge theory, especially the quantum chromodynamics, in which the (Wick-rotated Euclidean) spacetime is approximated by its discrete lattice. 
If we also assume that the spacetime is compact by imposing the periodic boundary condition, the lattice becomes a finite number of points. 
This approximation makes the path-integral finite-dimensional and hence calculable by e.g.~the Monte Carlo simulation. 

The main subject of this paper is the lattice approximation of the chiral anomaly in gauge theory. 
The chiral anomaly is the quantum breaking of the gauge symmetry, which the classical Yang--Mills theory possesses. 
According to Fujikawa's method, it is described by the Fredholm index of the Dirac operator twisted by a vector bundle.
Here a natural question arises: what is a lattice approximation of the chiral anomaly, or equivalently, the index of the twisted Dirac operator? 
A main difficulty for answering this question is that, although the Fredholm theory is essentially an infinite dimensional phenomenon, the approximated lattice Dirac operator acts on the finite dimensional Hilbert space of $\ell^2$-functions on the lattice. 

The answer is given in the literature of theoretical physics \cites{narayananConstructionLatticeChiral1995,hasenfratzIndexTheoremQCD1998} by considering the lattice (hermitian) Wilson--Dirac operator, a variation of the lattice Dirac operator defined by adding a new term called the Wilson term.
If the lattice scale is sufficiently fine, we obtain an invertible self-adjoint matrix by imposing a mass term to the hermitian Wilson-Dirac operator (this estimate is studied by Neuberger~\cite{neubergerBoundsWilsonDirac2000}). 
Then the half difference of the number of its positive and negative eigenvalues, instead of the number of the zero-modes, is a topological invariant. 

This integer actually coincides with the index of the Dirac operator at the continuum limit. 
This equality has been studied in the literature of theoretical physics \cites{luscherTopologyAxialAnomaly1999,suzukiSimpleEvaluationChiral1999,fujikawaContinuumLimitChiral1999,adamsAxialAnomalyTopological2002}. 
In this paper we shed new light on this index theorem from the viewpoint of higher index theory of almost flat vector bundles, i.e., a vector bundle equipped with a hermitian connection whose curvature is small. We mention that there are several other mathematical approaches to the same problem \cites{yamashitaLatticeVersionAtiyahSinger2020,fukayaAnalyticIndicesLattice}.

Another motivation of the paper is the study of almost commuting unitary matrices. Motivated from a question by Halmos~\cite{halmosUnsolvedProblemsUnknown1976}, Voiculescu~\cite{voiculescuAsymptoticallyCommutingFinite1983} constructs an example of a pair of unitary matrices $(U,V)$ such that the norm of the commutator $\| [U,V]\|$ is arbitrarily small but not perturbed to any commutative pair. 
Excel--Loring \cite{exelInvariantsAlmostCommuting1991} gives an alternative proof using a topological invariant motivated from C*-algebra K-theory. 
As a by-product of the proof of the lattice index theorem, we give an explicit formula of this invariant inspired from the hermitian Wilson--Dirac operator (Theorem \ref{thm:ACM}).

This topological invariant is related to the geometry of almost flat vector bundles by Connes--Gromov--Moscovici~\cite{connesConjectureNovikovFibres1990}.
A $d$-tuple of mutually almost commuting unitaries is thought of as a quasi-representation of the group $\bZ^d$. 
Generalizing the monodromy correspondence of flat bundles and representations, almost flat vector bundles on a space $M$ corresponds roughly in one-to-one with a quasi-representation of $\pi_1(M)$. 
It is clarified in the work of Hanke--Schick~\cite{hankeEnlargeabilityIndexTheory2006} and Dadarlat~\cite{dadarlatGroupQuasirepresentationsIndex2012} that the K-theoretic invariant of a group quasi-representation coincides with the index of Dirac operators twisted by the corresponding almost flat bundles.  This fact plays a key role in our proof of the index theorem of Wilson--Dirac operators.

\subsection{Statement of the theorem}
Let $d$ be a positive integer, standing for the dimension of the spacetime. 
For simplicity of discussion, here we only consider the case that $d$ is even (for a more general case, see Subsection \ref{section:4.3}). 
Let $V$ be the $d$-dimensional Euclidean space. That is, $V$ is the linear space $\bR^n$ equipped with the inner product $\langle{ \cdot}, {\cdot} \rangle$. Let $\Pi$ be the standard lattice of $V$, i.e., a cocompact rank $d$ free abelian subgroup $\Pi \subset V$ generated by an orthonormal basis $t_1,\cdots , t_n$ of $V$. Set $M:=V/\Pi$. It is equipped with the induced Riemannian metric, denoted by $g$. 

Let $\bCl(-V)$ denote the Clifford algebra of $V$ with the negative definite inner product. 
That is, $\bCl(-V)$ is the universal $\bC$-algebra generated by $\{ c(v) \mid v \in V\}$ with the relations $c(v)c(w)+c(w)c(v)=-2\langle v,w \rangle 1$. It is equipped with the C*-algebra structure by the adjoint $c(v^*)=-c(v)$.
Since $d$ is even, the algebra $\bCl(-V)$ has a unique irreducible $\ast$-representation $S$.

Let $E$ be a complex vector bundle on $M$ equipped with a hermitian inner product and let $\nabla$ be a hermitian connection on $E$. 
Then the twisted Dirac operator $D^E $ is defined as
\begin{align}
    D^E= \sum_{i=1}^n c(v_i) \nabla_{v_i} \colon \Gamma (M,E \otimes S) \to \Gamma (M , E \otimes S),\label{eq:Dirac}
\end{align}
where $v_1, \cdots , v_n$ is a standard basis of $T_0V \cong V$ corresponding to $t_1, \cdots, t_n$. Then $D^E$ is an odd self-adjoint operator with respect to the $\bZ_2$-grading of $S$. 
Let $D^E_0$ and $D^E_1$ denote the off-diagonal entries of $D^E$ with respect to the $\bZ_2$-grading $\gamma$ as $D_E=\begin{pmatrix} 0 & D^E_1 \\ D^E_0 & 0 \end{pmatrix}$.
Our main concern, the Fredholm index of $D^E$, is the integer
\[\Index (D^E) := \dim \ker (D^E_0) - \dim \ker (D^E_1),\]
which is well-defined since $D^E$ is an elliptic operator over a compact manifold $M$. 

Now we introduce a lattice approximation of the operator $D^E$. 
Firstly, we replace the covariant derivative $\nabla_{v_j}$ in (\ref{eq:Dirac}) with a difference operator. For a path $\gamma \colon [0,t] \to M$, we write $\Gamma_\gamma ^{E} \colon E_{\gamma(0)} \to E_{\gamma(1)}$ for the parallel transport of $(E,\nabla)$ along $\gamma$. Let $a >0$ stand for the scale of the lattice approximation (hence its is an inverse integer; $a=1/N$). Then the approximated spacetime is the quotient $a\Pi / \Pi \subset M$ of the finer lattice $a\Pi$.  
Now the shift operator in $j$-th direction is defined by the direct sum
\begin{align}
    U_j^{a,E}:=\bigoplus_{x \in a\Pi /\Pi} (\Gamma_{[x, x+a v_j]}^E \colon E_x \to E_{x+a v_j}), \label{eq:diff}
\end{align} 
which is a unitary operator on $\bigoplus_{x \in a\Pi / \Pi} E_x$. 
\begin{defn}\label{dfn:lattice}
The lattice covariant derivative $\nabla^{a,E}$, the lattice Dirac operator $D^{a,E}$, the lattice Wilson term $W^{a,E}$ and the lattice hermitian Wilson--Dirac operator $D^{a,E}_W$ with the scale $a=1/N$ are defined as following;
\begin{align}
\begin{split}
    \nabla^{a,E}_j &:= (U^{a,E}_j -1)/{a},\\
    D^{a,E}&:= \sum_j c(v_j) (\nabla^{a,E}_j - (\nabla^{a,E}_j)^*)/{2} ,\\
    W^{a,E}&:= \sum_j  (\nabla^{a,E}_j + (\nabla^{a,E}_j)^*)/{2} ,\\
    D_W^{a,E}&:= \gamma W^{a,E} + D^{a,E}.
\end{split}\notag 
\end{align}
Note that $D^{a,E}$, $W^{a,E}$ and $D_W^{a,E}$ acts on the Hilbert space 
\[ \sH^{a,E} :=\big(\bigoplus_{x \in a\Pi / \Pi} E_x \big) \otimes S.\]
\end{defn}

Morover, we impose the mass term $m_0\gamma$ to obtain the massive Wilson--Dirac operator $D_W^{a,E} + m_0 \gamma$. 
In this paper we consider the following two kinds of the mass term.
\begin{enumerate}
    \item The cut--off scale mass $m_0 = m/a$. 
    \item The constant mass term $m$ independent of the length-scale $a$.
\end{enumerate}
The massive hermitian Wilson--Dirac operator considered in the context of theoretical physics is mainly (1), or its functional calculus $(D_W^{a,E} + \frac{m}{a} )/ |D_W^{a,E} + \frac{m}{a}|$ (this operator is called the overlap Dirac operator).
On the other hand, the operator (2) is also considered in recent researches \cites{fukayaLatticeFormulationAtiyahPatodiSinger2019,fukayaAtiyahPatodiSinger2020}.
We show the index theorem for both of these operators. 

We define the "index" of the massive hermitian Wilson--Dirac operator as following. 
\begin{defn}\label{defn:I}
For a self-adjoint invertible matrix $A$, we define $\bfI(A)$ as
\[ \bfI (A):= \frac{\dim E_{>0}(A) - \dim E_{<0}(A)}{2}_{\textstyle ,} \]
where $E_{>0}(A)$ (resp.\ $E_{<0}(A)$) denotes the spectral subspaces of $A$ corresponding to positive (resp.\ negative) eigenvalues. 
\end{defn}
\begin{rmk}\label{rmk:I}
Assume that the (finite rank) Hilbert space on which $A$ acts is equipped with a $\bZ_2$-grading $\gamma$ such that the even and odd subspaces are isomorphic (the space $\sH^{a,E}$ above satisfies this assumption). Then we have 
\[\bfI (A) = \dim E_{>0}(A) - \dim E_{>0}(\gamma).  \]
This holds because $\dim E_{>0}(A) + \dim E_{<0}(A)=\dim \sH^{a,E}$ and $\dim E_{>0}(\gamma )=\dim \sH^{a,E}/2$.
\end{rmk}

\begin{thm}\label{thm:main}
The following hold.
\begin{enumerate}
    \item For any $0<m<2$, there is a length scale $a_0=1/N_0$ such that $D_W^{a,E} + \frac{m}{a}$ is invertible and 
    \[\bfI \Big(D_W^{a,E} +\frac{m}{a}\gamma\Big) = \Index (D^E) \]
    holds for any $0 < a < a_0$. 
    \item There is $m_0>0$ such that, for any $m>m_0$ there is a length scale $a_0=1/{N_0}$ such that the operator $D_W^{a,E} + m\gamma$ is invertible  and 
\[ \bfI (D_W^{a,E} +m\gamma) = \Index (D^E) \]
    holds for any $0<a<a_0$.
\end{enumerate}
\end{thm}

The strategy of the proof is as follows. 
Firstly, in Section \ref{section:2}, we relate the lattice shift operator (\ref{eq:diff}) with the quasi-representation $\pi_{a,E}$ of $\Pi = \pi_1(M)$ obtained by the constructions due to Gromov--Lawson~\cite{gromovPositiveScalarCurvature1983} and Connes--Gromov--Moscovici~\cite{connesConjectureNovikovFibres1990}. 
This enables us to understand the Wilson--Dirac operator $D^{a,E}_W$ as the image of the "universal" element $\hat{D}_W$ of the matrix-coefficient group algebra $\bC[\Pi] \otimes \End (S)$ through $\pi_{a,E}$, as is done in Section \ref{section:3}. In particular, here we show that $\hat{D}_W^{a,E}+m\gamma$ is an invertible self-adjoint operator representing the Bott element $\beta \in \K_0(C^*\Pi) \cong \K^0(\hat{\Pi})$.  
Finally, in Section \ref{section:4}, we relate $\pi_{a,E}(\beta)$ with the index of $D^E$ by applying the argument of Hanke--Schick~\cite{hankeEnlargeabilityIndexTheory2006}.

After proving Theorem \ref{thm:main}, in Section \ref{section:4.3}, we discuss two generalizations of this theorem. 
The first is the lattice family index, which has been considered by Adams~\cites{adamsFamiliesIndexTheory2002}.
The second is the lattice Real and Clifford equivariant index. 
This generalization is also studied in the forthcoming paper by Fukaya et.~al.~\cite{fukayaAnalyticIndicesLattice}. 

\subsection*{Acknowledgement}
The author would like to thank Mikio Furuta for his inspiring seminar talk on the Wilson--Dirac operator at RIKEN. This work was supported by RIKEN iTHEMS and JSPS KAKENHI Grant Numbers 19K14544, JPMJCR19T2, 17H06461.

\section{Lattice covariant derivative and group quasi-representation}\label{section:2}
In this section we relate the lattice covariant derivative $\nabla^{a,E}_j$ given in Definition \ref{dfn:lattice} with the quasi-representation of the group $\Pi$ obtained by a combination of the Gromov-Lawson construction and the almost monodromy correspondence. 

\subsection{Group quasi-representation and almost monodromy}
We start with a brief review of group quasi-representations, almost flat bundles and their almost monodromy correspondence. For a more detail on this subsection, we refer the reader to \cites{connesConjectureNovikovFibres1990,hankeEnlargeabilityIndexTheory2006,carrionQuasirepresentationsSurfaceGroups2013,dadarlatGroupQuasirepresentationsAlmost2014,kubotaAlmostFlatRelative2020}. 

In general, for a finitely presented discrete group $\Gamma$ and its finite set of generators $\cG$, we say that a map $\pi \colon \Gamma \to \mathrm{U}(P)$ is a \emph{$(\varepsilon , \cG)$-representation} of $\Gamma$ on $P$ if $\pi(e)=1$ and
\[ \|\pi(g)\pi(h)-\pi(gh)\|<\varepsilon   \]
for any $g,h \in \cG$ such that $gh \in \cG$. 
Our interest is quasi-representations of the group $\Pi \cong \bZ^d$ with respect to the generating set $\cG=\{ e, t_i^{\pm 1}, t_i^{\pm 1}t_j^{\pm 1} \mid i,j=1,\cdots ,n \} $. 
Then an $(\varepsilon , \cG)$-representation $(\pi,\sH)$ of $\Pi$ corresponds to a family of unitary operators $U_j:=\pi(t_j)$ which are mutually almost commutative, namely $\| [U_j,U_l] \| < 2\varepsilon$. Conversely, if we have a mutually $\varepsilon$-commutative $d$-tuple of unitaries, then $\pi(t_j^{\pm 1})=U_j^{\pm 1}$ and $\pi(t_j^{\pm 1}t_l^{\pm 1}):=U_j^{\pm 1}U_l^{\pm 1}$ for $j<l$ determines a $(2\varepsilon, \cG)$-representation of $\Pi$.  
A typical example of almost commuting unitaries in $d=2$ is
\[U_1 = 
\begin{pmatrix} 
1 & 0 & 0  &\cdots &0 \\ 
0 & \zeta_n &0 & \cdots& 0   \\ 
0 &0&\zeta_n^2& \cdots &0 \\ 
\vdots &\vdots  &\vdots & \ddots &\vdots  \\ 
0 & 0 & 0 & \cdots & \zeta_n^{n-1}  
\end{pmatrix}_{\textstyle ,} 
U_2 =\begin{pmatrix} 
0 & 0   &\cdots & 0 &1 \\ 
1 & 0  & \cdots& 0& 0   \\ 
0 &1& \cdots &\vdots & 0  \\ 
\vdots & \vdots  &\ddots & 0 &\vdots   \\  
0 & \cdots &  0 & 1 & 0 
\end{pmatrix}_{\textstyle ,} \]
where $\zeta_n$ denotes the root of unity $e^{2 \pi i /n}$.

There is a fruitful construction, established by Connes--Gromov--Moscovici in \cite{connesConjectureNovikovFibres1990}, of a quasi-representation of the fundamental group $\pi_1(M)$ of a manifold $M$ from an almost flat bundle on $M$ as the monodromy "representation". 
 A pair $(E, \nabla)$ is said to be a \emph{$(\varepsilon ,g)$-flat vector bundle} on $M$ if $E$ is a hermitian vector bundle on $M$ and $\nabla$ is a hermitian connection on $E$ whose curvature tensor $R^E \in \Omega ^2(M,\End E)$ satisfies
\[ \| R_E \|:= \sup _{x \in M} \sup _{\xi \in \bigwedge ^2 T_xM \setminus \{ 0 \} } \frac{\| R_E(\xi )\|_{\End (E_x) } }{\| \xi \| }  <\varepsilon.\]

As in Section \ref{section:1}, we denote by $\Gamma^\nabla_\gamma$ the parallel transport on $E$ along a path $\gamma \colon [0,1] \to M$ with respect to the connection $\nabla$. 
Let us choose a collection of closed loops $\ell_j \in \Omega M$ such that the set $\cG:=\{ [\ell_j]\} $ generates $\pi_1(M)$. 
\begin{defn}\label{dfn:mono}
The collection of operators 
\[\pi_E ([\ell _j]):= \Gamma_{\ell_j}^{\nabla} \colon E_0 \to E_0 \]
forms a $(C \varepsilon , \cG)$-representation for some $C>0$ depending only on a choice of $\{ \ell_j\} $ and $g$. We call this $\pi_E$ the almost monodromy quasi-representation of $E$.
\end{defn}
This definition makes sense because of the estimate
\[ \| \Gamma _\gamma^\nabla -1 \| \leq \| R^E\| \cdot \mathrm{Area} (D) \]
for any loop $\gamma$ bounded by a surface $D$, i.e., $\gamma =\partial D \subset M$ (we refer to \cite{gromovPositiveCurvatureMacroscopic1996}*{Section $4\frac{1}{4}$}). 
We call this quasi-representation as the almost monodromy quasi-representation of $E$.

Hereafter we focus on the torus $M=V/\Pi$ equipped with the standard Riemannian metric. 
We fix $\ell_j \in \Omega M$ as the geodesic loop in $v_j$-direction starting from the origin $0 \in V/\Pi$. Note that $\ell_j$ represents $t_j \in \pi_1(M)$. In this case $U_j:=\pi_E(t_j)$ satisfies $\| [U_j,U_l] \| < \| R_E\|$ since the loop $\ell_l^{-1} \circ \ell_j^{-1}\circ \ell_l\circ \ell_j$ is bounded by a standard square, which has area $1$.

\subsection{Gromov--Lawson construction}\label{section:2.2}
Gromov-Lawson \cite{gromovPositiveScalarCurvature1983} gives a systematic construction of almost flat bundles using the topology, i.e., enlargeability, of the base space. 
Here we review their construction by focusing on the case that the base space $M$ is a torus.

Let $\kappa_a \colon V \to V$ denote scaling map by $a$, that is, $\kappa _a(v)=av$ for any $v \in V$. 
Since $\kappa _a (N\Pi) = \Pi$, $\kappa_a$ induces a continuous map
\[\kappa _a \colon V/N\Pi \to V/\Pi.\]
Note that $\kappa_a$ is $a^2$-area contracting, that is, the induced bundle map 
\[ d\kappa_a \colon \lwedge^2 TM \to \lwedge ^2TM\]
satisfies $ d\kappa_a = a^2$. Hence we have
\[\| R_{\kappa_a^*E} \| =  a^2\| R_E\| . \]
Let $q_a \colon V/N\Pi \to V/\Pi$ denote the $N^n$-fold covering map. Then the push-forward bundle
\[E_N:= q_{a!}\kappa_a^*E = \bigsqcup_{x \in V/\Pi} \bigoplus _{q(\bar{x})=x} (\kappa_a^*E)_{\bar{x}} \]
has a connection induced from that of $\kappa_a^*E$. 

\begin{lem}\label{lem:index}
For any vector bundle $E$ on $V/\Pi$, we have
\[ \Index (D^{E_N}) = \Index (D^E).\]
\end{lem}
\begin{proof}
By definition $q_!$ extends to a unitary 
\[ q_{a!} \colon L^2(V/N\Pi , \kappa_a^*E) \to L^2(V/\Pi,q_{a!}\kappa_a^*E).\]
Also, the conformal transformation $\kappa_a$ gives rise to a unitary 
\[ \kappa_a^* \colon L^2(V/\Pi, E) \to L^2(V/N\Pi, \kappa_a^*E) \]
determined by $\kappa_a^*(\xi)(x)=\frac{1}{a^{d/2}} \xi(\kappa_a(x))$. Now the diagram
\[
\xymatrix{
L^2(V/\Pi,q_!\kappa_a^*E) \otimes S\ar[r]^{D^{E_N}} \ar[d] & L^2(V/\Pi, q_!\kappa_a^*E)\otimes S \ar[d] \\
L^2(V/N\Pi , \kappa_a^*E) \otimes S \ar[r]^{D^{\kappa_a^*E}} \ar[d] & L^2(V/N\Pi , \kappa_a^*E) \otimes S  \ar[d] \\ 
L^2(V/\Pi,E) \otimes S \ar[r]^{\frac{1}{a}D^{E}} & L^2(V/\Pi,E) \otimes S
}
\]
commutes, and hence $\Index (D^{E_N}) = \Index (\frac{1}{a}D^E)=\Index (D^E)$ holds.
\end{proof}

\subsection{Lattice covariant derivative as almost monodromy}\label{section:2.3}
Finally we observe that the lattice covariant derivatives $U_j^{a,E}$ introduced in Section \ref{section:1} is nothing but the almost monodromy quasi-representation of the almost flat bundle obtained by the Gromov-Lawson construction.

Let $\bv_1 , \cdots , \bv_n$ denote the standard basis of the vector space $V$. 
Let $\ell_j$ denote the image of the affine segment $[0,\mathbf{v}_j] \subset V$ to $V/\Pi$, which is a closed loop representing $t_j \in \pi_1(V/\Pi)$. 
Let $\pi_{a,E}$ denote the corresponding almost monodromy quasi-representation of the bundle $q_!\kappa _a^* E$ in the sense of Definition \ref{dfn:mono}. Then it is a $( \| R_E\| a^2,\cG)$-representation of $\Pi$.
\begin{lem}\label{lem:GL}
The $(a^2\| R_E\| ,\cG)$-representation $\pi_{a,E}$ satisfies $\pi_{a,E}(t_j)=U_j^{a,E}$.
\end{lem}
\begin{proof}
The fiber of $E_N =q_{a!}\kappa_a^* E$ at $0 \in V/\Pi$ is identified with the direct sum 
\[ \bigoplus_{q_a(x)=0} (\kappa_a^*E)_x = \bigoplus_{x \in \Pi/N\Pi} (\kappa_a^*E)_x =\bigoplus_{x \in a\Pi/\Pi} E_x . \]
The connection on $E_N$ is imposed from that of $E$ in the way that the parallel transport $\Gamma_{\gamma}^{E_N}$ along a path $\gamma \colon [0,1] \to V/\Pi$ is written as
\begin{align*}
 \Gamma_{\gamma}^{E_N} &= \bigoplus_{q_a(\tilde{\gamma}) =\gamma} \big( \Gamma_{\tilde\gamma}^{\kappa_a^*E} \colon \kappa_a^*E_{\tilde{\gamma}(0)} \to \kappa_a^*E_{\tilde{\gamma}(1)}  \big) \\
 &=\bigoplus_{q_a(\tilde{\gamma}) =\gamma} \big( \Gamma_{\kappa_a \circ \tilde\gamma}^{E} \colon E_{\kappa_a(\tilde{\gamma}(0))} \to E_{\kappa_a(\tilde{\gamma}(1))}  \big).
\end{align*}
In particular we have
\begin{align*}
    \pi_{a,E}(t_j)=\Gamma_{[0,\bv_j]}^{E_N} &= \bigoplus_{x \in \Pi/N\Pi} \big( \Gamma_{[x,x+\bv_j]}^{\kappa_a^*E} \colon (\kappa_a^*E)_{x} \to (\kappa_a^*E)_{x+\bv_j}  \big)\\
    &=\bigoplus_{x \in a\Pi/\Pi} \big( \Gamma_{[x,x+a\bv_j]}^{E} \colon E_{x} \to E_{x+a\bv_j}  \big).
\end{align*}
The right hand side is the same as the operator $U_j^{a,E}$ defined in (\ref{eq:diff}).
\end{proof}

\section{K-theory of the universal hermitian Wilson--Dirac operator}\label{section:3}
In the last section we observe that the shift operators $U_j^{a,E}$ are viewed as the image of $t_j \in \Pi$ under the quasi-representation $\pi_{a,E}$. This fact suggests that the lattice Wilson--Dirac operator is also viewed as the image of the `universal' element lying in the group algebra $\bC [\Pi]$, or its C*-algebra completion. 
In this section we describe the universal Wilson--Dirac operator as a self-adjoint element of a C*-algebra. This relates the invariant $\bfI(D_W^{a,E} + m\gamma)$ introduced in Definition \ref{defn:I} with the Bott element of the C*-algebra K-group $\K_0(C^*\Pi)$.

\subsection{Universal hermitian Wilson--Dirac operator}
Let $\hat{\Pi}$ denote the Pontrjagin dual $\Hom(\Pi,\bT)$ of $\Pi$. According to the Gelfand--Naimark duality, the group C*-algebra $C^*\Pi$ is isomorphic to the C*-algebra $C(\hat{\Pi})$ of continuous functions on $\Pi$ by regarding $t \in \Pi$ as the continuous function $\chi \mapsto \chi(t)$. 
\begin{defn}\label{defn:WD}
We define the universal hermitian Wilson--Dirac operator $\hat{D}_W \in C^*\Pi \otimes \End(S) $ as $\hat{U}_j:=t_j$, $\hat \nabla_j:= \hat U_j-1$ and 
\begin{align*}
    \hat{D}&:= \sum c(v_j)({\hat \nabla_j - \hat \nabla_j^*})/{2}, \\
    \hat W&:= \sum ({\hat \nabla_j + \hat \nabla_j^*})/{2},\\
    \hat{D}_W&:=\hat{D} +\gamma \hat{W}.
\end{align*} 
\end{defn}
Let $U_1, \cdots, U_d$ be a $d$-tuple of mutually $\varepsilon$-commuting unitaries on a (finite rank) Hilbert space $\sH$ and let $\pi$ denotes the $(\varepsilon , \cG)$-representation of $\Pi$ determined by $\pi(t_j):=U_j$ and $\pi(t_j^{\pm 1}t_l^{\pm 1}):=U_j^{\pm 1}U_l^{\pm 1}$. We extend $\pi$ to a linear map $\pi \colon \bC[\Pi] \to \bB(\sH)$ by 
\begin{align}
\pi (t_1^{k_1} \cdots t_d^{k_d}) :=\pi(t_1)^{k_1}\cdots \pi^\flat(t_d)^{k_d}.
\end{align}
Then we can associate to $\pi$ a matrix
\begin{align} \pi(\hat{D}_W) = \sum_{j=1}^n \frac{1}{2}(U_j - U_j^*) \cdot c(v_j) + \sum_{j=1}^n \Big( \frac{1}{2}(U_j+U_j^*) -1\Big) \cdot \gamma . \end{align}

By comparing the above definitions with Definition \ref{dfn:lattice}, Lemma \ref{lem:GL} is rephrased as the following lemma.
\begin{lem}\label{lem:universality}
The hermitian Wilson--Dirac operator $D^{a,E}_W$ coincides with $\frac{1}{a}\pi_{a,E} (\hat{D}_W)$.
\end{lem}

Since the elements $\hat{U}_j$ and $\hat{U}_l$ are commutative, the spectral theory of the square $(\hat{D}_W +m\gamma )^2$ of the massive universal hermitian Wilson--Dirac operator is much easier than analysing Wilson--Dirac operators itself.  
\begin{lem}\label{lem:invertible}
For $0<m<2$, the massive universal hermitian Wilson--Dirac operator $\hat{D}_W + m\gamma$ is invertible.
\end{lem}
\begin{proof}
Here we write as $t_j=e^{2\pi i k_j}$, where $k_j$ is the $j$-th coordinate of $\hat{\Pi} \cong (\bR/\bZ)^d$. Through the identification $C^*\Pi \cong C(\hat{\Pi})$, the operators $\hat{D}$ and $\hat{W}$ are identified with $\End(S)$-valued functions on $\hat{\Pi}$ as
\[\hat{D}=\sum_{j=1}^d c(v_j) \cdot i \sin (2\pi k_j), \ \ \ \ \hat{W}=\sum_{j=1}^d (\cos(2\pi k_j)-1).\]
Hence we have
\begin{align*}
    \begin{split}
    &(\hat{D}_W + m\gamma)^2 = \hat{D}^2 + (\hat{W}+m)^2\\
    =&\sum _j\sin (2 \pi k_j)^2 + \Big( \big( \sum_j \cos(2\pi k_j)-1 \big) +m \Big) ^2 \geq 0. 
    \end{split}
\end{align*}
The first and second components are both non-negative and never be zero simultaneously. Therefore $(\hat{D}_W + m\gamma)^2 >0$ holds, i.e., the matrix-valued function $\hat{D}_W + m\gamma$ is invertible.
\end{proof}

\subsection{K-theory class of the universal Wilson--Dirac operator}
Let $A$ be a C*-algebra. 
If $A$ is unital, i.e., $1 \in A$, its $\K_0$-group $\K_0(A)$ is defined as the set of homotopy classes of projections in $\bM_n(A)$ with the summation given by the direct sum (for a foundation of C*-algebra K-theory, we refer to \cite{rordamIntroductionTheoryAlgebras2000}).
We remark that the space of self-adjoint and invertible operators in a C*-algebra $A$ is homotopy equivalent to the space of projections by the continuous map $h \mapsto (h/|h| +1)/2$, with the homotopy inverse $p \mapsto 2p-1$. Hereafter, for a self-adjoint operator $h \in \bM_n(A) $ we simply write $[h]$ for the corresponding element of the $\K_0$-group $[(h/|h| +1)/2] \in \K_0(A)$.

By the Serre--Swan theorem, the $\K_0$-group $\K_0(C(X))$ of the commutative C*-algebra $C(X)$ of continuous functions on a compact space $X$ is isomorphic to the topological $\K$-group $\K^0(X)$. Hence the group $\K_0(C^*\Pi) \cong \K^0(\hat{\Pi})$ contains the Bott element $\beta$ of top degree, i.e., the image of the generator of $\K^0(\bD^n,S^{n-1}) \cong \bZ$ to $\K^0(\hat{\Pi})$ with respect to an open embedding.
\begin{lem}\label{lem:Bott}
Let $c(v)$ and $\gamma$ be elements of $\bCl_d \cong \End(S)$ as in Section \ref{section:1}. Set
\[h:= \sum_{j=1}^N c(v_j) \cdot i x_j + \gamma x_0 \colon S^n \to \End(S),\]
where $x_0,\cdots, x_n$ is the standard coordinate of $S^n \subset \bR^{n+1}$.
Then $h$ is a self-adjoint matrix-valued function with $h^2=1$ and the Bott element $\beta \in \K^0(S^n)$ is represented by $h$ as $\beta =[h]-[\gamma ]$.  
\end{lem}
\begin{proof}
This is understood from the viewpoint of Segal's connective K-theory \cite{segalKhomologyTheoryAlgebraic1977}.
Let $F_d(n)$ denote the set of mutually commuting $d+1$-tuple $(A_0,A_1, \cdots , A_d)$ of $n \times n$ matrices such that $\sum A_i^2=1$ and set $F_d(\infty)=\bigcup_n F_d(n)$ with respect to the inclusion $(A_0,A_1,\cdots, A_n) \mapsto (A_0 \oplus 1, A_1 \oplus 0, \cdots, A_d \oplus 0)$. 
Similarly, let $\Phi_d(n)$ denote the set of self-adjoint unitaries on $S \otimes \bC^n$ and set $\Phi_d(\infty) = \bigcup_n \Phi_d(n)$ with respect to $B \mapsto B \oplus \gamma$. 
Then the map $\varphi \colon F_d(\infty) \to \Phi_d(\infty)$ given by $(A_0, A_1, \cdots , A_d) \mapsto A_0\gamma + \sum i A_j c(v_j)$ induces an isomorphism of the $\pi_d$-group (\cite{segalKhomologyTheoryAlgebraic1977}*{Proposition 1.3}). Moreover, the map $c \colon F_d \to P_d$, where $P_d$ is the configuration space of points in $S^d$ (in other words, $P_d=\mathrm{SP}^\infty(S^d)$), sending $(A_0,A_1, \cdots, A_d)$ to its joint spectrum, also induces an isomorphism of the $\pi_d$-group (\cite{kubotaJointSpectralFlow2016}*{Proposition 3.4}). Now the map $\bx:=(x_0,x_1, \cdots ,x_d) \colon S^d \to F_d(1)$ satisfies $\varphi \circ \bx = h$ and $[c \circ \bx ] \in [S^d,P_d] \cong H^d(S^d) \cong \bZ$ is the generator. This shows the lemma. 
\end{proof}

As is proved in Lemma \ref{lem:invertible}, the massive hermitian Wilson--Dirac operator $\hat{D}_W + m \gamma$ is a self-adjoint invertible element of the C*-algebra $C^*\Pi \otimes \End(S)$, and hence it determines an element $[\hat{D}_W+m\gamma] \in \K_0(C^*\Pi) \cong \K^0(\hat{\Pi})$.
 
\begin{prp}\label{lem:univ}
For any $0 < m < 2$, the K-theory class $[\hat{D}_W + m\gamma ] \in \K_0(C^*\Pi) \cong \K^0(\hat{\Pi})$ satisfies
    \[[\hat{D}_W + m\gamma ] - [\gamma ] = \beta. \]
\end{prp}
\begin{proof}
Recall the identification of $\hat{D}_W$ with an $\End(S)$-valued continuous function on $\hat{\Pi}$ as in the proof of Lemma \ref{lem:invertible}. 
Let 
\[ f:=\Big(\sum _j\sin (2 \pi k_j)^2 + \big( \big( \sum_j \cos(2\pi k_j)-1 \big) +m \big) ^2 \Big)^{1/2}_{\textstyle ,}\]
which is a positive function on $\hat{\Pi}$. In the space of self-adjoint invertible operator-valued functions, $\hat{D}_W + m\gamma$ is homotopic to 
\[ \frac{\hat{D}_W + m\gamma}{f} = \sum c(v_j)\frac{i\sin (2\pi k_j)}{f} + \gamma \frac{\big( \sum_j (\cos(2\pi k_j)-1) \big) + m }{f}=F^*h, \] 
where $F \colon \hat{\Pi} \to S^d$ is defined as 
\[ F(\bk)=\bigg( \frac{\big( \sum_j (\cos(2\pi k_j)-1) \big) + m }{f(\bk)}_{\textstyle ,} \frac{\sin(2\pi k_1)}{f(\bk)}_{\textstyle ,} \cdots {\bigg. }_{\textstyle ,} \frac{\sin (2 \pi k_d)}{f(\bk)}\bigg)\]
for $\bk=(k_1,\cdots, k_d) \in \hat{\Pi}$. Since $(1,0,\cdots , 0)$ is the regular value of $F$ and $F^{-1}(1,0,\cdots ,0)=\{ (0,\cdots, 0)\}$, the degree of $F$ is $1$ and hence 
\[ [\hat{D}_W + m\gamma] - [\gamma] = F^*([h]-[\gamma])=\beta \in \K^0(\hat{\Pi})\]
by Lemma \ref{lem:Bott}.
\end{proof}

\begin{rmk}
The invertible matrix-valued function $\hat{D}_W + m\gamma$ and its K-theory class in $ \K^0(\hat{\Pi})$, particularly when $d =0,1,2,3$, is also known as a model of Chern insulators in the theory of topological insulators in condensed-matter physics (see for example \cite{prodanBulkBoundaryInvariants2016}*{Section 2.2.4}).  
\end{rmk}

\subsection{Approximate $\ast$-homomorphism and K-theory}\label{section:3.3}
For a unital C*-algebra $A$, a unitary representation $\pi \colon \Pi \to \cU(A)$ extends to a $\ast$-homomorphism $C^*\Pi \to A$. Although a single quasi-representation does not necessarily extend to a continuous map on $C^*\Pi$, a nice collection $\pi_n$ of quasi-representations is able to be treated as a single $\ast$-homomorphism in the following way. 

Let $\pi_n \colon \cG \to \cU(\sH_n)$ be a sequence of finite rank $(\varepsilon_n,\cG)$-representations such that $\varepsilon_n \to 0$. We pick a collection of embeddings $\sH_n \subset \sH$ into a single separable infinite dimensional Hilbert space $\sH$. We simply write $\bK$ for the compact operator algebra $\bK(\sH)$. Let us consider the quotient C*-algebra
\[ \cQ:= \frac{\prod _{n \in \bN} \bK }{\bigoplus _{n \in \bN} \bK}_{\textstyle .} \]
For $(a_n)_{n \in \bN} \in \prod _{\bN} \bK$, we write $(a_n)_{n \in \bN}^{\flat }$ for its image in $\cQ$.
denote the projection. Then
\begin{align} \pi^\flat (t_j):= (\pi_n(t_j) \oplus 1_{\sH_n^\perp})_{n \in \bN}^\flat  \label{eq:pitilde} \end{align}
satisfies $[\pi^\flat(t_j) , \pi^\flat(t_l)] =0$. That is, $\pi^\flat$ extends to a unitary representation of $\Pi $ by $\pi^\flat (t_1^{k_1} \cdots t_d^{k_d}) :=\pi^\flat(t_1)^{k_1}\cdots \pi^\flat(t_d)^{k_d}$, and hence we obtain a $\ast$-homomorphism 
\[\pi^\flat \colon C^*\Pi \to \cQ. \]

\begin{lem}\label{lem:KofQ}
The $\K_0$-group of C*-algebras $\bigoplus_\bN \bK$, $\prod_\bN \bK$ and $\cQ$ are 
\[ \K_0({\textstyle \bigoplus_\bN \bK} ) \cong {\textstyle \bigoplus_\bN \bZ}, \ \ \ \K_0({\textstyle \prod_\bN \bK}) \cong {\textstyle \prod_\bN \bZ}, \ \ \ \K_0 (\cQ) \cong \frac{\prod_\bN \bZ }{\bigoplus_\bN \bZ} \]
respectively. 
\end{lem}
\begin{proof}
The first and second isomorphisms are straightforward from the definition of C*-algebra K-theory. The third is a consequence of the six-term exact sequence associated to the C*-algebra extension $0 \to \bigoplus_\bN \bK \to \prod_\bN \bK \to \cQ \to 0$. 
\end{proof}
Hereafter, for a sequence $(k_n)_{n \in \bN}$ of integers, we write $(k_n)_{n \in \bN}^\flat$ for its under the quotient $\prod_\bN \bZ \to \prod_\bN \bZ / \bigoplus_\bN \bZ$.

By the above lemma, $\pi^\flat$ induces a group homomorphism
\[\pi^\flat \colon \K_0(C^*\Pi) \to \K_0 (\cQ) \cong \frac{\prod_\bN \bZ }{\bigoplus_\bN  \bZ}_{\textstyle .}\]
In other words,
\begin{align}
     (\pi_n)_{n \in \bN} \mapsto \pi^\flat (\beta) \in  \frac{\prod_\bN \bZ }{\bigoplus_\bN  \bZ} \label{eq:Bott}
\end{align}
is a topological invariant for sequences of $(\varepsilon_n,\cG)$-representations of $\Pi$ with $\varepsilon_n \to 0$. It is trivial if $(\pi_n)_{n \in \bN}$ is homotopic to a representation, i.e., $\pi^\flat$ lifts to a $\ast$-homomorphism $C^*\Pi \to \prod_\bN \bK$ since any $\ast$-homomorphism $C(\hat{\Pi}) \to \bK$ is homotopic to the trivial one. 

Now we apply Proposition \ref{lem:univ} to get an explicit description of the invariant (\ref{eq:Bott}).  

\begin{lem}\label{lem:est}
Let $\lambda$ denote the bottom of the spectrum of $(\hat{D}_W + m\gamma)^2$ and let $0 < C< \lambda$.
Then there is a constant $\varepsilon >0$ such that $\pi(\hat{D}_W + m\gamma)^2>C$ holds for any $(\varepsilon , \cG)$-representation $\pi$ of $\Pi$. 
\end{lem}
Instead of an analytic proof using norm estimates (as is given in Proposition \ref{prp:apriori} later), we prove this lemma by using an abstract spectral theory of C*-algebras. 
\begin{proof}
Let us choose $C'$ with $C < C' < \lambda$ and we show that $\pi_n(\hat{D}_W+m\gamma)^2 \geq C'$ except for finitely many $n$'s. To this end, we assume the contrary, i.e., for any $n \in \bN$ there is a $(\frac{1}{n}, \cG)$-representation $\pi_n$ of $\Pi$ such that $\pi(\hat{D}_W + m\gamma)^2-C'$ is not positive. 
We bundle these $\pi_n$'s to get a $\ast$-homomorphism $\pi^\flat$ as in (\ref{eq:pitilde}). 
Then the spectrum $\sigma (\pi^\flat(\hat{D}_W^2 +m\gamma)^2-C')$ is the set of accumulation points of $\sigma (\pi_n(\hat{D}_W + m\gamma)^2-C')$, and hence it intersects with $[-C',0]$ non-trivially. 
This contradicts with 
\begin{align*}
    \sigma (\pi^\flat(\hat{D}_W + m\gamma )^2-C') &= \sigma (\pi^\flat((\hat{D}_W + m\gamma )^2-C')) \\ 
    &\subset \sigma ((\hat{D}_W + m\gamma )^2-C') \subset [\lambda - C',\infty)
\end{align*} 
which is a standard fact in the spectral theory of C*-algebras.
\end{proof}

This lemma shows that  
\begin{align} 
\bfI(\pi(\hat{D}_W) + m \gamma ) \in \bZ
\end{align}
is a topological obstruction for $\pi  $ to be homotopic to a representation of $\Pi$ in the space of $(\varepsilon , \cG)$-commuting matrices for sufficiently small $\varepsilon >0$. 

\begin{thm}\label{thm:ACM}
Let $\pi_n$ be a sequence of $(\varepsilon _n ,\cG)$-representations of $\Pi$ with $\varepsilon_n \to 0$ such that $\pi_n(\hat{D}_W + m\gamma)$ is invertible for $n \geq n_0$. Then we have
\[ \pi^\flat (\beta) = (\bfI(\pi_n(\hat{D}_W) + m\gamma ) )_{n \geq n_0}^\flat  \in \frac{\prod \bZ}{\bigoplus \bZ}_{\textstyle .}\]
\end{thm}
Here we write $(\cdot)_{n \geq n_0}^\flat$ for the image under the compositions $\prod_{n \geq n_0} \bK \to \prod _{n \in \bN} \bK \to \cQ$ or $\prod _{n \geq n_0} \bZ \to \prod _{n \in \bN}\bZ \to \prod _\bN \bZ/ \bigoplus_\bN \bZ$.
\begin{proof}
The identification $\K_0(\cQ) \cong \prod _\bN \bZ/ \bigoplus_\bN \bZ$ is given by mapping $[(p_n)_{n \in \bN}^\flat] \in \K_0(\cQ)$ to $( \rank p_n)_{n \in \bN}^\flat$. Hence, for an invertible element $(h_n)_{n \in \bN}^\flat \in \cQ$, its K-theory class $[(h_n)_{n \in \bN}^\flat] \in \K_0(\cQ)$ corresponds to $( \dim E_{>0}(h_n) )_{n \in \bN}^\flat$.

By Proposition \ref{lem:univ} we have
\begin{align*}
\pi^\flat(\beta)&= [\pi^\flat(\hat{D}_W + m\gamma)] - [\gamma] \\
&=([ ( \pi_n(\hat{D}_W + m\gamma) )_{n \geq n_0}^\flat ] - [\gamma].
\end{align*}
Through the identification as above, the right hand side corresponds to 
\[ ( \dim E_{>0}(\pi_n(\hat{D}_W + m\gamma)) - \dim E_{>0}(\gamma))_{n \in \bN}^\flat = ( \bfI (\pi_n(\hat{D}_W) +m\gamma ))_{n \geq n_0} ^\flat \]
by Reamrk \ref{rmk:I} and Lemma \ref{lem:universality}.
\end{proof}

\begin{rmk}
When $d=2$, this invariant is the same as the one given in  \cite{exelInvariantsAlmostCommuting1991}. Our description using the Wilson--Dirac operator is efficient in computation since there is no need to calculate the functional calculus of a large matrix. 
\end{rmk}

\section{Wilson--Dirac index theorem as almost flat index pairing}\label{section:4}
In this section we give a proof of Theorem \ref{thm:main} as an application of the Hanke-Schick index pairing \cite{hankeEnlargeabilityIndexTheory2006}. 
\subsection{Higher index pairing}
The index pairing is the pairing between the K-theory and K-homology groups of a manifold $M$ (a reference is \cite{higsonAnalyticHomology2000}). It assigns to $[\sE] \in \K^0(M)$ and $ [\sD] \in \K_0(M)$,  where $\sE$ is a complex vector bundle and $\sD$ is an elliptic operator on $M$,  the Fredholm index $\Index \sD^\sE$ of the twisted operator $\sD^\sE$. 
More generally, a K-homology element $[\sD] \in \K_0(M)$ induces the "index pairing with coefficient", i.e., a homomorphism 
\[\langle \blank , [\sD] \rangle_A  \colon \K_0(C(M) \otimes A) \to \K_0(A).\]

This homomorphism is described as following. Let $\sH$ be a Hilbert space equipped with a $\ast$-representation $C(M) \to \bB(\sH)$. A K-homology element in the Atiyah--Kasparov picture (see for example \cite{higsonAnalyticHomology2000}*{Definition 8.1.1}) is represented by a bounded operator $F \in \bB(\sH)$ such that $[F,f], F^*F-1, FF^*-1$ are compact operators for any $f\in C(M)$. 
Our concern is the Dirac fundamental class $[D]$ of $M$ represented by $\sH:=L^2(M,S^+) \cong L^2(M,S^-)$ (this unitary isomorphism is given by a fixed Borel isomorphism of vector bundles $S^+ \cong S^-$) and 
\[ F:=D(1+D^*D)^{-1/2} \colon L^2(M,S^+) \to L^2(M,S^-).\]

Another ingredient of the proof is the Mishchenko bundle $\cL:= \tilde{M} \times_\Pi C^*\Pi$, where $\Pi$ acts on $C^*\Pi$ by the multiplication from the left. This is a flat bundle of Hilbert $C^*\Pi$-modules. Note that the Serre--Swan theorem with coefficient also holds; the group $\K_0(C(M) \otimes A)$ is isomorphic to the Grothendieck group of the semigroup of bundles of finitely generated projective Hilbert $A$-modules (see for example \cite{schickIndexTheoremsKKtheory2005}*{Subsection 3.2}). Therefore, $\cL$ determines an element $[\cL] \in \K_0(C(M) \otimes C^*\Pi)$. 
\begin{rmk}\label{rmk:pair}
Here we enumerate some basic facts on the index pairing and the Mishchenko bundle which will be used in the proof of Theorem \ref{thm:main}. 
\begin{enumerate}
    \item The index pairing with coefficient in $\bC$ is the same thing as the usual index pairing, i.e., $\langle [E], [D] \rangle _\bC =\Index (D^E) \in \K_0(\bC) \cong \bZ$ (see for example \cite{higsonAnalyticHomology2000}*{Proposition 4.8.10 (c)}).
    \item The index pairing is compatible with the base-change: Let $\phi \colon A \to B$ be a $\ast$-homomorphism. Then we have
\begin{align*}
    \langle (\id_{C(M)} \otimes \phi)_* [p], [F] \rangle_B = \phi_*(\langle [p], [F] \rangle_A). 
\end{align*}
    This holds by definition of the index pairing above because $(\id_{\bB(\sH)} \otimes \phi)(F_p)=F_{\phi(p)}$ holds.
    \item The bundle $\cL$ is universal among flat bundles of Hilbert C*-modules: If we have a flat bundle $\cE \to M$ of Hilbert $A$-modules, the monodromy representation $\pi_\cE \colon \Pi \to \bB(E)$ (where $E$ is a typical fiber of $\cE$) is associated. Then, by definition of $\cL$, the base change $\cL \otimes_{\pi_E} E$ is isomorphic to $\cE$ itself. In the level of K-theory, we have 
    \[ (\id_{C(M)} \otimes  \pi_{\cE})_*[\cL] = [\cE].\] 
    \item The higher index pairing $\langle [\cL], [D] \rangle_{C^*\Pi} \in \K_0(C^*\Pi)$ is the Bott element $\beta $. This is essentially proved by Lusztig~\cite{lusztigNovikovHigherSignature1972} by using the Atiyah-Singer family index theorem. Indeed, through the identification $C^*\Pi \cong C(\hat{\Pi})$, the index pairing with $\cL$ is identified with the family index for the fiberwise Dirac operator on the fiber bundle $M \times \hat{\Pi} \to \hat{\Pi}$ twisted by the Poincare line bundle 
    \[\cP:=V \times \hat{\Pi} \times \bC /\{ (v,\chi, \xi) \sim (v+t, \chi, \chi(t)^*\xi)\}.\]
\end{enumerate}
\end{rmk}
\begin{lem}[{\cite{hankeEnlargeabilityIndexTheory2006}*{Theorem 3.8}}]\label{lem:HS}
Let $\pi_N:=\pi_{a,E}$ be the quasi-representation of $\Pi$ defined in Subsection \ref{section:2.3} and let $\pi^\flat$ be the $\ast$-homomorphism as in (\ref{eq:pitilde}). Then we have
\[\pi^\flat(\beta) = ( \Index D^E )_{N \in \bN} ^\flat \in \frac{\prod \bZ}{\bigoplus \bZ}_{\textstyle .} \]
\end{lem}
\begin{proof}
For the self-consistency of the paper, we quickly review the proof given in \cite{hankeEnlargeabilityIndexTheory2006}. 
For each $N$, $\cE_N:=E_N \otimes \sH^*$ is a bundle of projective Hilbert $\bK(\sH)$-modules which is identified with $E_N$ through the isomorphism $\K_0(C(M)) \cong \K_0(C(M) \otimes \bK(\sH)))$. The connection of $E_N$ extends to that of $\cE_N$ whose curvature has the norm less than $a^2 \| R_E\|$ (for the connection and curvature of Hilbert C*-module bundles, see \cite{schickIndexTheoremsKKtheory2005}*{Section 4}) . Let $\cE:=\prod \cE_N$, which is a Hilbert $\prod \bK$-module. Then we have
\[\langle [\cE ], [D] \rangle_{\prod \bK}  = ( \Index D^{E_N} )_N = ( \Index (D^E) )_N \in \K_0(\prod \bK) \cong \prod \bZ \]
by Lemma \ref{lem:index} and Remark \ref{rmk:pair} (1). Moreover, the base-change $\cE \otimes _\varphi \cQ$ is isomorphic to the flat bundle $V \times_{\pi^\flat} \cQ $ associated to $\pi^\flat$. 

Let $\eta^\flat$ denote the quotient map $\prod \bK \to \cQ$. Now we obtain that
\begin{align*}
\pi^\flat_*(\beta ) =&\pi^\flat_*(\langle [\cL], [D] \rangle_{C^*\Pi }) \\
=& \langle (\id_{C(M)} \otimes \pi^\flat)_*[\cL], [D] \rangle_{\cQ }\\
=&\langle (\id _{C(M)} \otimes \eta^\flat)_*[\cE], [D] \rangle_{\cQ } \\
=& \eta_*^\flat(\langle [\cE], [D] \rangle_{\prod \bK } ) \\
=& ( \Index (D^E) )_{N \in \bN}^\flat \in \frac{\prod \bZ}{\bigoplus \bZ} _{\textstyle .} 
\end{align*}
Here each of the first, second, third and forth equalities follows from (4), (2), (3), (2) of Remark \ref{rmk:pair} respectively. 
\end{proof}

\subsection{Proof of the main theorem}\label{section:4.2}
The proof of Theorem \ref{thm:main} (1) is essentially finished by Lemma \ref{lem:est}, Theorem \ref{thm:ACM} and Lemma \ref{lem:HS}. For the proof of (2), we need a quantitative refinement of Lemma \ref{lem:est}. 

The following norm estimate is communicated to the author by Mikio Furuta. A more conceptual description from the viewpoint of geometric analysis will appear in the forthcoming paper \cite{fukayaAnalyticIndicesLattice}. 
A similar estimate is also studied by Neuberger~\cite{neubergerBoundsWilsonDirac2000}. 
\begin{prp}\label{prp:apriori}
Let $D_{\kappa, m}:= \kappa \pi_{a,E}(\hat{D}_W) + m\gamma$. For any $\kappa \in [m,1/a]$, the square $D_{\kappa, m}^2$ is bounded from below by $m^2- 4 d^2\| R_E\|$.
\end{prp}
\begin{proof}
In the proof, we write $a \sim_\varepsilon b$ for $\| a - b\| < \varepsilon$. 
Set 
\begin{align*}
     x_j&:=\frac{1}{2i}(U_j^{a,E}- (U_j^{a,E})^*), \\
     y_j&:=\frac{1}{2}(U_j^{a,E}+ (U_j^{a,E})^*), \\
     z_j&:=U_j^{a,E}-1.
\end{align*}
Note that we have $1-y_j = \frac{1}{2}z_jz_j^*$ and $\| z_j\| \leq 2$. Now the square $D_{\kappa,m}^2$ is bounded from below as
\begin{align*}
    D_{\kappa, m}^2 &= \big( \kappa \sum_j c_j \cdot i x_j + \gamma \big( \kappa \sum_j (y_j-1) + m \big) \big)^2\\
    &\sim _{\varepsilon_1} \kappa^2 \sum_j x_j^2 + \big( \kappa \sum_j (y_j-1) +m\big)^2 \\
    &=\kappa^2 \sum_j x_j^2+ \kappa^2\sum_j (y_j-1)^2 + \sum_{j \neq l} (1-y_j)(1-y_l)\\
    & \ \ \  \ + 2\kappa m \sum_j (y_j-1) + m^2\\
    &= 2\kappa (\kappa -m) \sum_j (1-y_j) + \kappa ^2 \sum_{j \neq l}(1-y_j)(1-y_l) + m^2\\
    &\sim_{\varepsilon_2} \kappa (\kappa -m) \sum_j z_j^*z_j + \frac{\kappa ^2}{4}\sum_{j \neq l}z_jz_j^*z_l^*z_l + m^2 \geq m^2,
\end{align*}
where $\varepsilon_1>0$ and $\varepsilon_2>0$ are defined as 
\begin{align*}
     &\Big\| -\sum_{j < l} \kappa^2c_jc_l[x_j,x_l] +  \sum_{j} c_j\gamma  \Big[ \kappa \cdot i x_j,\big( \kappa \sum_l (y_l-1) + m \big) \Big] \Big\|\\
     \leq & \kappa ^2\sum_{j,l} (\| [x_j,x_l]\| + \| [x_j,y_l\|) \leq \kappa^2 \cdot 2d^2 \| R_E\|a^2  \leq 2d^2  \| R_E\| =:\varepsilon_1
\end{align*}
and
\begin{align*}
     &\frac{\kappa ^2}{4} \Big\| \sum_{j \neq l} (z_j[z_j^*,z_l]z_l^* + z_jz_l[z_j^*,z_l^*]) \Big\| \leq 2d^2 \kappa ^2 \| R_E\|a^2  \leq 2d^2  \| R_E\| =:\varepsilon_2.
\end{align*}

This inequality implies that the spectrum of $D_{\kappa, m}^2$ is bounded below by $m^2 - 4 d^2  \| R_E\| $ if $m \leq  \kappa \leq 1/a$. 
\end{proof}
As a consequence of Proposition \ref{prp:apriori}, we obtain that the self-adjoint matrices $\pi_{a,E}(m\hat{D}_{W})+m\gamma$ and $D_W^{a,E} +m\gamma$ are homotopic in the space of invertible self-adjoint matrices for any $m>4d^2\| R_E\|$. In particular we obtain
\begin{align}
    \bfI (\pi_{a,E}(\hat{D}_{W}) + \gamma )  = \bfI (D^{a,E}_W + m\gamma).  \label{eq:mass}
\end{align}

\begin{proof}[Proof of Theorem \ref{thm:main}]
By (\ref{eq:mass}), it suffices to show that 
\[ \bfI  (\pi_{a,E}(\hat{D}_W) + \gamma ) = \Index (D^E) \]
for sufficiently small $a>0$.
This follows from Theorem \ref{thm:ACM} and Lemma \ref{lem:HS}.
\end{proof}

\section{Generalizations}\label{section:4.3}
In this section we discuss two generalizations of Theorem \ref{thm:main}, the family version and the real and Clifford equivariant version. They are also considered by Adams~\cites{adamsFamiliesIndexTheory2002} and Fukaya et.~al.~\cite{fukayaAnalyticIndicesLattice} respectively. 
Since our method is $\K$-theoretic and does not rely on any evaluation such as integration of differential forms, the same proof also works for these generalized setting.  
\subsection{Family index}
Let $X$ be a compact space and let $E $ be a vector bundle over $M \times X$, which is thought of as a family of vector bundles $E_x :=E|_{M \times \{ x\}}$ parametrized by $X$. 
Let $\{ \nabla_x \}_{x \in X}$ be a smooth family of hermitian connections on $E_x$.
Then the fiberwise Dirac operator $\fD^E(x):= D^{E_x}$ is a continuous function on $X$ taking value in Fredholm operators. Its family index determines an element of the topological K-group $\K^0(X)$. 

Also, the Gromov--Lawson construction as in Section \ref{section:2} provides a continuous family $\pi_{a,E,X}:=\{ \pi_{a,E_x} \}_{x \in X}$ of quasi-representations of $\Pi$ parametrized by $X$. It gives rise to a $\ast$-homomorphism $\pi^\flat_X \colon C^*\Pi  \to \cQ \otimes C(X)$. Note that the K-group $\K_0(\cQ \otimes C(X))$ is isomorphic to $\prod \K^0(X) / \bigoplus \K^0(X)$. 
The family Wilson--Dirac operator $\fD_W^{a,E} :=\pi^\flat_X(\hat{D}_W)$ is a continuous function from $X$ to the space of self-adjoint matrices such that $\fD_W^{a,E} +m\gamma$ is invertible. 
Then $E_{>0}(\fD_W^{a,E}+m\gamma )$ is a vector bundle on $X$, and 
\[ \bfI_X(\fD_{W}^{a,E}+m\gamma ):=[E_{>0}(\fD_W^{a,E}+\textstyle { \frac{m}{a}}\gamma ) ]- [E_{>0}(\gamma)] \in \K^0(X) \]
is defined.  

Now the same proof shows that both Theorem \ref{thm:ACM} and Lemma \ref{lem:HS} hold for this family version. Consequently we obtain the following generalization. 
\begin{thm}
The following hold.
\begin{enumerate}
    \item For any $0<m<2$, there is $a_0=1/N_0$ such that $\fD_W^{a,E} + \frac{m}{a}$ is invertible and $\bfI_X (\fD_W^{a,E} +\frac{m}{a}\gamma ) = \Index_X (\fD^E)$ holds for any $0 < a < a_0$. 
    \item There is $m_0>0$ such that, for any $m>m_0$ there is $a_0=1/{N_0}$ such that $\fD_W^{a,E} + m\gamma$ is invertible and $\bfI_X (\fD_W^{a,E} +m\gamma) = \Index_X (\fD^E)$ holds for any $0<a<a_0$.
\end{enumerate}
\end{thm}

\subsection{Real and Clifford-equivariant index}
Here we consider the index theorem for real vector bundles on the torus $M$ with an arbitrary dimension. Throughout this subsection $\Cl_{p,q}$ denotes the Clifford algebra generated by $e_1, \cdots, e_p$ and $f_1, \cdots, f_q$ with $e_j^2=1$ and $f_j^2=-1$. 
Let $M$ be the $d$-dimensional standard torus and let $\cS $ be the unique irreducible representation of the Clifford algebra $\Cl(V \oplus -V) \cong \Cl_{d,d}$.  
The $\Cl_{d,0}$-Dirac operator 
\[\cD:=\sum _{j=1}^n f_j \nabla_{v_j} \colon \Gamma(M,\cS ) \to \Gamma(M,\cS) \]
is an odd self-adjoint operator anticommuting with the Clifford generators $e_j$ of $\Cl_d$. We write $\cD^E$ for the corresponding twisted $\Cl_{d,0}$-Dirac operator. 
In a similar way as in \cite{atiyahIndexTheorySkewadjoint1969}*{Section 5}, its Clifford index is defined as
\[ \Index_{\Cl_{d,0}} (\cD^E):= [\ker \cD^E, \gamma ] \in \hat{\fM}_d / i(\hat{\fM}_{d+1}) \cong \KO_{d}. \]
Here $\hat{\fM}_d$ denotes the set of pairs $(W,h)$, where $W$ is a finite dimensional representation of $\Cl_{d,0}$ and $h \in \End(W)$ is a $\bZ_2$-grading on $W$ anticommuting with Clifford generators of $\Cl_{d,0}$, and let $i \colon \hat{\fM}_{d+1} \to \hat{\fM}_{d}$ denotes the map forgetting the action of $d+1$-th Clifford generator. 
Note that the real $\K$-group $\KO_{d}$ is isomorphic to one of $\bZ$, $\bZ_2$ or $0$. 

On the other hand, the quasi-representation $\pi_{a,E}$ obtained by the Gromov--Lawson construction for $E$ is also real, i.e., each $\pi_{a,E}(t_j)$ is a real orthogonal matrix. 
In the same way as (\ref{eq:Dirac}), we define the hermitian $\Cl_{d,0}$-Wilson--Dirac operator $\cD_\cW^{a,E}$ as
\[\cD_\cW^{a,E} = \cD^{a,E} + \gamma \cW^{a,E} \in \bB(\sH_{a,E}) \]
where $\sH_{a,E}:=\bigoplus_{x \in a\Pi/\Pi} E_{x} \hotimes \cS $. 
Then the massive hermitian $\Cl_{d,0}$-Wilson--Dirac operator $\cD_\cW^{a,E}+m\gamma $ is an invertible self-adjoint operator anticommuting with Clifford generators $e_j$.
Hence it determines an element
\[\bfI_{\Cl_{d,0}}(\cD_\cW^{a,E}+m\gamma):=[\sH_{a,E}, h_{a,E} ] - [\sH_{a,E} , \gamma] \in \hat{\fM}_{d}/i_*\hat{\fM}_{d+1}, \]
where $h_{a,E} =(\cD_\cW^{a,E} +m\gamma)/|\cD_\cW^{a,E} +m\gamma| $.

We will extend the argument given in Section \ref{section:3} and \ref{section:4} to this setting, we introduce a generalization of Karoubi's definition of the real $\K$-group to Real C*-algebras (cf.\ \cite{kubotaNotesTwistedEquivariant2016}*{Corollary 5.15}). Here we say that a Real C*-algebra is a C*-algebra $A$ equipped with an antilinear $\ast$-isomorphic involution $a \mapsto \overline{a}$. 
For example, let $X$ be a compact Real space, i.e., $X$ is a compact space equipped with an involution $\tau \colon X \to X$. Then $A=C(X)$ with $\bar{f}(x)=\overline{f(\tau (x))}$ is a Real C*-algebra. 
Let $\Delta_d$ denote the direct sum of all $\bZ_2$-graded irreducible representation of $\Cl_{d,0}$ with the $\bZ_2$-grading $\gamma$. 
\begin{defn}
Let $A$ be a unital Real C*-algebra. Let us define the set
\[\sF_n^{d}(A):=\{ s \in A \hotimes \bK(\Delta_d^n) \mid s=s^*, \bar{s}=s, s^2=1, e_js=-se_j \text{ for $j=1, \cdots, d$} \}. \]
We define the Real $\K$-group $\KR_d(A)$ as the set of homotopy classes of $\bigcup _n\sF_n^d(A)$, where $\sF_n^d(A) \subset \sF_{n+1}^d(A)$ is defined by $s \mapsto s \oplus \gamma $. The summation is given by the direct sum and the zero element is represented by $\gamma $.
\end{defn}

\begin{rmk}\label{rmk:KR}
We shortly give some remarks on the above definition of the $\KR$-group.
\begin{enumerate}
\item When $d=0$, $\Delta_d$ is the $\bZ_2$-graded $\bR$-vector space $\bR \oplus \bR^{\mathrm{op}}$. The map $[s] \mapsto [(s+1)/2]-[(\gamma +1)/2]$ gives an isomorphism of $\KR_0(A)$ and the Grothendieck group of Real (i.e., invariant under the involution) projections in $\bigcup_{n}\bM_n(A)$.
\item For a Real space $X$, the group $\KR_d(C(X))$ is isomorphic to the Real K-group $\KR^{-d}(X)$ defined in \cite{karoubiTheory2008}*{Exercise III.7.14}. In particular, when $A=\bR$, the group $\KR_{d}(\bR)$ is isomorphic to $\KO_{d}$ by mapping $[s]$ to $[\Delta_d^n , s]$. 
\item The Real K-theory extends to the Kasparov theory. In particular, the index pairing with the Real K-homology cycle $[\cD] \in \KR^{-d}(C(M))$ induces a map
\[\langle {\cdot},  [\cD] \rangle_A \colon \KR_0(C(M) \otimes A) \to \KR_d(A). \]
This index pairing also satisfies Remark \ref{rmk:pair} (1), (2) and (3) by assuming $\phi$ and $\cE$ to be Real.
\end{enumerate}
\end{rmk}
If $\sH$ is a real Hilbert space, then $\bK:=\bK(\sH \otimes _\bR \bC)$ is equipped with a canonical Real C*-algebra structure. In the same way as Subsection \ref{section:3.3} we define the Real C*-algebra $\cQ:=\prod_\bN \bK / \bigoplus_\bN \bK$. Then $(\pi_{a,E})_{N \in \bN}$ gives rise to a Real $\ast$-homomorphism $\pi^\flat \colon C^*\Pi \to \cQ$, where the Real structure on the group C*-algebra $C^*\Pi$ determined by $\bar{t_j} = t_j$. Note that it is checked in the same way as Lemma \ref{lem:KofQ} that $\KR_d(\cQ)$ is isomorphic to $\prod \KO_{d} / \bigoplus \KO_{d}$. 

The statement analogous to Lemma \ref{lem:Bott} and Remark \ref{rmk:pair} (4) holds for this Real setting. Note that  is identified with that on $C(\hat{\Pi})$ induced from the involution $\tau(\chi) = -\chi$. We write $\bR^{p,q}$ for the Real space $\bR^p \oplus i\bR^q$, $\bD^{p,q}$  for the unit disk of $\bR^{p,q}$ and $S^{p,q}$ for the unit sphere of $\bR^{p,q}$. 
\begin{lem}\label{lem:real}
The following hold.
\begin{enumerate}
    \item The function $h:=\sum f_j \cdot i x_j + \gamma x_0 \colon S^{1,d} \to \End(\cS)$ is self-adjoint, invertible, Real and anticommutes with the Clifford generators $e_j$. Then $[h] =\beta \in \KR_{d}(C(S^{1,d}))$.
    \item The Mishchenko bundle $\cL:=\tilde{M} \times_\Pi C^*\Pi$ is a Real bundle of finitely generated Hilbert $C^*\Pi$-modules on $M$. Moreover,  $\langle [\cL], [\cD] \rangle_{C^*\Pi} = \beta \in \KR^d(\hat{\Pi})$. 
\end{enumerate}
\end{lem}
Note that Lemma \ref{lem:real} (1) is consistent with Lemma \ref{lem:Bott} because $[h]$ in Lemma \ref{lem:real} corresponds to the element written as $[h]-[\gamma]$ in Lemma \ref{lem:Bott} through the isomorphism in Remark \ref{rmk:KR} (1).
\begin{proof}
The proof of Lemma \ref{lem:Bott} shows that the map $ \KR_d(C(S^{1,d})) \to \K_d(C(S^d))$ forgetting the Real involution on $C(S^{1,d})$ sends $[h]$ to the complex Bott generator $\beta \in \K_d(C(S^d))$. This shows (1) since the above forgetful map in this degree is an isomorphism.

Next we show (2). As is stated in \cite{mathaiTDualitySimplifiesBulk2016}*{Section 6}, particularly the equation (6.8), the index pairing $\langle [\cL], {\cdot } \rangle _{C^*\Pi} $ respects the stable splittings $\KO_d(M) \cong \bigoplus_{l} \KO_{d}(\bD^{l,0} , S^{l,0})^{\binom d l}$ and $\KR^{-d}(\hat{\Pi} ) \cong \bigoplus \KR^{-d}(\bD^{0,l}, S^{0,l})^{\binom d l}$. In particular, it sends $[\cD]$ lying in the top degree summand of $\KO_d(M)$ to the top degree summand of $\KR^{-d}(\hat{\Pi})$. This shows that $\langle [\cL],[\cD] \rangle _{C^*\Pi} = k\beta $ for some $k \in \bZ$. To see that $k=1$, consider the image of $\langle [\cL],[\cD] \rangle _{C^*\Pi}$ to the complex K-group by the forgetful map and remind Remark \ref{rmk:pair} (4).
\end{proof}

Remark \ref{rmk:KR} and Lemma \ref{lem:real} are enough to check that the same argument as in Sections \ref{section:3} and \ref{section:4} also work in this setting. Finally we obtain the following theorem. 

\begin{thm}\label{thm:real}
The following hold.
\begin{enumerate}
    \item For any $0<m<2$, there is $a_0=1/N_0$ such that $\cD_\cW^{a,E} + \frac{m}{a}$ is invertible and $\bfI_{\Cl_{d,0}} (\cD_\cW^{a,E } +\frac{m}{a}\gamma ) = \Index_{\Cl_{d,0}} (\cD^E )$ holds for any $0 < a < a_0$. 
    \item There is $m_0>0$ such that, for any $m>m_0$ there is $a_0=1/{N_0}$ such that $\cD_\cW^{a,E} + m\gamma$ is invertible and $\bfI_{\Cl_{d,0}} (\cD_\cW^{a,E } +m\gamma) = \Index_{\Cl_{d,0}} (\cD^E )$ holds for any $0<a<a_0$.
\end{enumerate}
\end{thm}

\begin{rmk}
The index pairing with a quaternionic vector bundle, instead of a real vector bundle, is also treated in a similar way. 
In this case, $\Index _{\Cl_{d,0}}(\cD^E)$ and $\bfI_{\Cl_{d,0}}(\cD^{a,E}_\cW +m\gamma)$ take value in the group of quaternionic representations of the Clifford algebra $\Cl_{d,0}$, which is isomorphic to $\KO_{d+4}$.

Indeed, let us identify the quaternion field $\bH$ with its complexification $\bH \otimes _\bR \bC$ as the Real C*-algebra. It is isomorphic to $\bM_2(\bC)$ equipped with the Real involution $a \mapsto u\bar{a}u^*$, where $\bar{a}$ is the usual complex conjugation and $u=\big( \begin{smallmatrix}0 & 1 \\ 1 & 0 \end{smallmatrix} \big) $. 
Then a quaternionic vector bundle $E$ determines an element of $\KR_0(C(M) \otimes \bH )$ and the sequence $(\pi_{a,E})_{N \in \bN}$ gives rise to a Real $\ast$-homomorphism $\pi^\flat \colon C^*\Pi \to \cQ \otimes \bH$. Hence the same proofs of Theorem \ref{thm:ACM} and Lemma \ref{lem:HS} work only by replacing $\cQ$ with $\cQ \otimes \bH$. 
\end{rmk}

\bibliographystyle{alpha}
\bibliography{ref.bib}

\end{document}